\documentclass[11pt, a4paper]{article}
\usepackage{graphicx}
\usepackage{mathptmx}
\usepackage{hyperref}
\usepackage{physics}
\usepackage{latexsym,amsmath,amssymb,amsbsy,amsthm,graphicx,euscript}

\newtheorem{lemma}{Lemma}
\newcommand{\hil}{\mathcal H}
\newcommand{\scr}{\scriptscriptstyle}
\newcommand{\lsc}[1]{_{\scr #1}}

\newcommand{\lph}[1]{\lambda_{#1}^{(\scriptscriptstyle\Phi_{\mu})}}
\newcommand{\lphf}[1]{\lambda_{#1}^{(\scriptscriptstyle\Phi)}}
\newcommand{\lps}[1]{\lambda_{#1}^{(\scriptscriptstyle\Psi)}}
\newcommand{\cv}[1]{c^{(\scriptscriptstyle #1)}}
\newcommand{\cvd}[1]{c^{(\scriptscriptstyle #1)\dagger}}

\newcommand{\spv}{\sup_{\substack{\scriptscriptstyle\Psi\in V\\\scriptscriptstyle\norm{\Psi}=1}}}
\newcommand{\lsup}{\bar\lambda^{\scriptscriptstyle V}_1}

\begin{document}



\title{Lower bounds on concurrence and negativity from a trace inequality
}

\author{\footnotesize K. V. ANTIPIN,\\
	Physics Department, Lomonosov Moscow State University,\\
        Moscow, Russia,\\
        kv.antipin@physics.msu.ru.
}


\maketitle


\begin{abstract}
For bipartite quantum states we obtain lower bounds on two important entanglement measures, concurrence and negativity, studying the inequalities for the expectation value of a projector on some subspace of the Hilbert space. Several applications, including analysis of stability of entanglement  under various perturbations of a state, are discussed.

\end{abstract}


\section{Introduction}	

Entanglement as a resource is a central concept in quantum information theory. The important question is to tell whether a given quantum composite system state is entangled or separable. One of the first remarkable results in this direction was the positive partial transposition~(PPT) criterion~\cite{Peres} as necessary condition for separability of  bipartite mixed states. This simple but extremely useful observation by Asher Peres has generated further considerable research. It was proved that PPT condition is necessary and sufficient  for separability of $2\otimes2$ and $2\otimes3$ states~\cite{Horod1}. Over time several other necessary or/and sufficient criteria were developed~\cite{Horod1, Terhal_Bell,Horod2, Realign, CCN}, among which entanglement witnesses~\cite{Terhal_Bell, Horod2} and the CCNR criterion~\cite{Realign, CCN} proved to be important tools in detecting entanglement.

At the same time, considerable research was devoted to developing various entanglement measures~\cite{GuRev}. Among the most important measures are the concurrence~\cite{EntFrac} and the negativity~\cite{Zyc,wneg}. Some of the known results connect these entanglement measures  to specific separability criteria via inequalities giving lower bounds on these quantities. For example, the connection between the concurrence of a bipartite $m\otimes n$ state $\rho$ and the PPT and realignment criteria was established in Ref.~\cite{coppt}:
\begin{equation}
	C(\rho )\geq \sqrt{\frac{2}{m(m-1)}}\Big(\max (\Vert \rho ^{T_{A}}\Vert
	,\Vert \mathcal{R}(\rho )\Vert )-1\Big),  \label{pptreal}
\end{equation}
where $\norm{}$ -- trace norm, $\rho ^{T_{A}}$ -- partial transposition of $\rho$ with respect to subsystem $A$, and $\mathcal{R}(\rho )_{ij,kl}=\rho _{ik,jl}$.

Recently~\cite{Bloch} lower bounds on the two entanglement measures were obtained with the use of Bloch repfesentations.

In the present paper we obtain lower bounds on concurrence and negativity studying inequalities for the expectation value of a specific operator in a state described by a given density operator. The operator is chosen to be a projector $\Pi_V$ on some subspace $V$ of the Hilbert space. The main result of this paper is the following inequality
\begin{equation}
    C(\rho\lsc{AB})\,\geqslant\,\max\left(\sqrt{\frac{2}{m(m-1)}}\,\frac{\Tr{\rho\lsc{AB}\,\Pi_V} - \lsup}{\lsup},\,0\right)
\end{equation}
for the concurrence of a bipartite $m\otimes n$ state $\rho\lsc{AB}$ and the inequality
\begin{equation}
    N^{\mathrm{CREN}}(\rho\lsc{AB})\,\geqslant\,\max\left(\frac{\Tr{\rho\lsc{AB}\,\Pi_V} - \lsup}{2\lsup},\,0\right)
\end{equation}
for the convex-roof extended negativity~(defined further). Here $\lsup$ --- the supremum of the largest Schmidt coefficient squared taken over all vector states in the subspace $V$.

The essential tool in the derivation of the bounds is the von Neumann's trace inequality~\cite{Neumann} which relates the trace of a product of two matrices with their singular values. Obtained results generalize some known estimates for the entanglement measures. The underlying separability criterion was investigated earlier in Ref.~\cite{SpecProp}.

\section{Definitions}

Throughout this paper we consider bipartite pure and mixed states.

The \emph{concurrence} of a pure bipartite state $\psi$ is defined as follows:
\begin{equation}\label{purconc}
	    C(\psi) = \sqrt{2\left(1-\Tr{\rho_A^2}\right)},
    \end{equation}
    where $\rho=\dyad{\psi}$ is the density operator of the pure state $\ket{\psi}$, and $\rho_A = \mathrm{Tr}_B\{\rho\}$ -- the reduction of $\rho$ on subsystem $A$.

    Given the Schmidt decomposition $\ket{\psi} = \sum_i\,\sqrt{\lambda_i}\ket{\alpha_i\beta_i}$, the concurrence can be expressed in the following way~\cite{ElToSie2015}:
    \begin{equation}\label{schconc}
	    C(\psi) = 2\,\sqrt{\sum_{i<j}\,\lambda_i\lambda_j}.
    \end{equation}

    A mixed state $\rho$ can be expressed via various  ensemble decompositions of the form
\begin{equation}
	\rho = \sum_a\,p_a\,\dyad{\psi_a}
\end{equation}
    By definition, the concurrence of a mixed state $\rho$ is given by the convex roof construction, the minimum average concurrence taken over all ensemble decompositions of $\rho$:
    \begin{equation}\label{concconv}
	        C(\rho) = \min_{\{(p_a,\,\psi_a)\}}\,\sum_a\,p_a\,C(\psi_a).
	\end{equation}

	The \emph{negativity} of $\rho$ is defined as follows:
	\begin{equation}\label{negdef}
		    N(\rho) = \frac12(\norm{\rho^{T_B}}_1 - 1),
	    \end{equation}
	    where $\rho^{T_B}$ is the partial transpose of $\rho$ with respect to party $B$, and $\norm{A}_1 = \Tr{\sqrt{A^{\dagger}A}}$ is the trace norm of $A$.

	    $N(\psi)$ can also be expressed in terms of the Schmidt coefficients~\cite{ElToSie2015}:
	    \begin{equation}\label{pureneg}
		        N(\psi) = \sum_{i<j}\,\sqrt{\lambda_i\lambda_j}.
		\end{equation}

		From the definition of $N$ it is seen that entanglement of states with a positive partial transpose~(PPT states) is not detected by this measure.

		The \emph{convex-roof extended negativity~\cite{CREN2003}}~(CREN) is given by
		\begin{equation}\label{cren}
			    N^{\mathrm{CREN}}(\rho) = \min_{\{(p_a,\,\psi_a)\}}\,\sum_a\,p_a\,N(\psi_a).
		    \end{equation}
		     
		     The convex roof measures presented in Eqs.~(\ref{concconv}) and (\ref{cren}) detect \emph{all} entangled states, but they are very hard to compute, so lower and upper bounds on these quantities play important role in entanglement theory.

\section{Derivation of bounds}
We proceed to the lower bounds on the concurrence and the convex-roof extended negativity.

Let $\rho\lsc{AB}$ be a bipartite density operator acting on a tensor product of Hilbert spaces $\hil_A\otimes\hil_B$ with $\dim{\hil_A} = m\,\leqslant\dim{\hil_B} = n$. Let 
\begin{equation}\label{anydec}
	\rho_{\scr AB}=\sum_{\mu} q_{\mu} \dyad{\Phi_{\mu}}\lsc{AB},
\end{equation}
be an  ensemble decomposition of $\rho_{\scr AB}$ with ensemble probabilities $q_{\mu}$ and vector states $\ket{\Phi_{\mu}}$.

Let $V$ be an $l$-dimensional subspace  of $\hil_A\otimes\hil_B$ spanned by orthonormal vectors $\{\ket{\Psi_k}\}$, $k=1,\,\ldots,\,l$. Consider the expectation value of $\Pi_V = \sum_{k}\,\dyad{\Psi_k}$, the projector on $V$, in the state $\rho_{\scr AB}$:
\begin{equation}\label{TrProj}
		\Tr{\rho\lsc{AB}\,\Pi_V} = \Tr{\sum_{\mu} q_{\mu} \dyad{\Phi_{\mu}}\,\sum_{k}\,\dyad{\Psi_k}}
		=\sum_{\mu}q_{\mu}\sum_k\absolutevalue{\bra{\Phi_{\mu}}\ket{\Psi_k}}^2.
\end{equation}
Now let us look at the expression $\sum_k\absolutevalue{\bra{\Phi_{\mu}}\ket{\Psi_k}}^2$ on the right-hand side of Eq.~(\ref{TrProj}). Let $\ket{\Psi}$ be some vector state in $V$ with decomposition
\begin{equation}
	    \ket{\Psi} = \sum_k\,c_k\,\ket{\Psi_k},\:\:\sum_k\,\absolutevalue{c_k}^2 = 1.
    \end{equation}
    Vector $\ket{\Phi_{\mu}}$ can be decomposed as
    \begin{equation}
	        \ket{\Phi_{\mu}} = \sum_k\,t_k\ket{\Psi_k}\:+\:\ket{\Phi_{\mu}^{\perp}},\:\:\:\sum_k\,\absolutevalue{t_k}^2\,\leqslant\,1,
	\end{equation}
	where $\ket{\Phi_{\mu}^{\perp}}$ belongs to the orthogonal complement of $V$.

Next, by the Cauchy-Schwarz inequality we have:
\begin{multline}\label{cssup}
    \absolutevalue{\bra{\Phi_{\mu}}\ket{\Psi}}^2 = \absolutevalue{\sum_k\,c_k\,\bra{\Phi_{\mu}}\ket{\Psi_k}}^2\,\leqslant\\
    \leqslant\,\sum_j\,\absolutevalue{c_j}^2\,\sum_k\absolutevalue{\bra{\Phi_{\mu}}\ket{\Psi_k}}^2 = \sum_k\absolutevalue{\bra{\Phi_{\mu}}\ket{\Psi_k}}^2.
\end{multline}
Let us assume\footnote{otherwise, inequality in Eq.~(\ref{cssup}) is trivial} that $\sum_k\,\absolutevalue{t_k}^2 \ne\,0$ and let $\alpha$ denote a number such that $\absolutevalue{\alpha}^2\sum_k\,\absolutevalue{t_k}^2 = 1$. If $\ket{\Phi_{\mu}}$ is fixed, the upper bound in Eq.~(\ref{cssup}) is achieved and the inequality becomes equality when a specific vector $\ket{\Psi}$ is chosen:
\begin{equation}
    \ket{\Psi} = \sum_k\,\alpha t_k\,\ket{\Psi_k},
\end{equation}
i. e., $\ket{\Psi}$ is proportional to the projection of vector $\ket{\Phi_{\mu}}$ onto subspace $V$. Consequently, we can write:
\begin{equation}
    \sum_k\absolutevalue{\bra{\Phi_{\mu}}\ket{\Psi_k}}^2 = \max_{\substack{\scriptscriptstyle\Psi\in V\\\scriptscriptstyle\norm{\Psi}=1}}\,\absolutevalue{\bra{\Phi_{\mu}}\ket{\Psi}}^2,
\end{equation}
and Eq.~(\ref{TrProj}) transforms to
\begin{equation}\label{Trmax}
    \Tr{\rho\lsc{AB}\,\Pi_V} = \sum_{\mu}\,q_{\mu}\,\max_{\substack{\scriptscriptstyle\Psi\in V\\\scriptscriptstyle\norm{\Psi}=1}}\,\absolutevalue{\bra{\Phi_{\mu}}\ket{\Psi}}^2.
\end{equation}

Let $c^{(\scriptscriptstyle\Phi_{\mu})}$ and $ c^{(\scriptscriptstyle\Psi)}$ denote the matrices of the vectors $\ket{\Phi_{\mu}}$ and $\ket{\Psi}$ respectively, given in the computational basis of $\hil_A\otimes\hil_B$:
\begin{equation}
	  \label{PureDec}
	  \ket{\Psi} = \sum_{i, j} c^{(\scriptscriptstyle\Psi)}_{ij}\ket{i}_{\scr A}\otimes\ket{j}_{\scr B}, \:\:\: \ket{\Phi_{\mu}} = \sum_{i, j}  c^{(\scriptscriptstyle\Phi_{\mu})}_{ij}\ket{i}_{\scr A}\otimes\ket{j}_{\scr B}.
    \end{equation}
    Eq.~(\ref{Trmax}) can be rewritten as:
    \begin{equation}
    	\label{Trmatr}
    \Tr{\rho\lsc{AB}\,\Pi_V} = \sum_{\mu}\,q_{\mu}\,\max_{\substack{\scriptscriptstyle\Psi\in V\\\scriptscriptstyle\norm{\Psi}=1}}\,\absolutevalue{\Tr{c^{(\scriptscriptstyle\Phi_{\mu})\dagger}\,c^{(\scriptscriptstyle\Psi)}}}^2.
    \end{equation}
We can give an upper bound on the trace on the right-hand side of Eq.~(\ref{Trmatr}) with the use of the property~\cite{Neumann,MA,Bhatia} that is often referred to as the von Neumann's trace inequality:
\begin{equation}
  \label{TrSing}\
  \lvert\mathrm{Tr}\{A^{\dagger}B\}\rvert \leqslant \sum_{i=1}^{q} \sigma_i (A) \sigma_i (B),
\end{equation}
where $A, B$ -- complex $m\times n$ matrices, $q = \mathrm{min}\{m,n\}$, $\sigma_i(A), \,\sigma_i(B)$ -- singular values of A and B arranged in non-increasing order: $\sigma_1(A)\geqslant\sigma_2(A)\geqslant\ldots\geqslant\sigma_q(A)$.

Let us denote as $\sqrt{\lph i}$ and $\sqrt{\lps i}$ the Schmidt coefficients of the vectors $\ket{\Phi_{\mu}}$ and $\ket{\Psi}$ respectively, arranged in non-increasing order: $\sqrt{\lph 1}\,\geqslant\,\sqrt{\lph 2}\,\geqslant\,\ldots\,\geqslant\sqrt{\lph m}\:\:$; $\sqrt{\lps 1}\,\geqslant\,\sqrt{\lps 2}\,\geqslant\,\ldots\,\geqslant\sqrt{\lps m}$. By definition of the Schmidt coefficients, the relation with the singular values of $\cv{\Phi_{\mu}}$ and $\cv{\Psi}$ is as follows:
\begin{equation}\label{sing}
	\sqrt{\lph i} = \sigma_i(\cv{\Phi_{\mu}}),\: \sqrt{\lps i} = \sigma_i(\cv{\Psi}).
\end{equation}
Making use of Eqs.~(\ref{TrSing}), (\ref{sing}), we obtain the following chain of inequalities:
\begin{multline}\label{TrBound}
	\absolutevalue{\Tr{\cvd{\Phi_{\mu}}\cv{\Psi}}}^2\,\leqslant\,\left(\sum_i\,\sqrt{\lps i\lph i}\right)^2\,\leqslant\\
	\leqslant\,\left(\sum_i\,\sqrt{\lph i}\right)^2\,\lps 1\,\leqslant\,\left(\sum_i\,\sqrt{\lph i}\right)^2\,\spv\,\lps 1
\end{multline}
Let us denote $$\lsup \equiv \spv\,\lps 1.$$
Eqs.~(\ref{Trmatr}) and (\ref{TrBound}) yield:
\begin{multline}\label{mainbound}
    \Tr{\rho\lsc{AB}\,\Pi_V}\, \leqslant\, \lsup\,\sum_{\mu}\,q_{\mu}\,\left(\sum_i\,\sqrt{\lph i}\right)^2=\\
    =\,\lsup\,\sum_{\mu}\,q_{\mu}\,\left(\sum_i\,\lph i + 2\sum_{i<j}\,\sqrt{\lph i\lph j}\right)=\\
    =\,\lsup\left(1 + 2\sum_{\mu}q_{\mu}\sum_{i<j}\,\sqrt{\lph i\lph j}\right).
\end{multline}
From Eqs.~(\ref{pureneg}), (\ref{cren}), and (\ref{mainbound}) we obtain a lower bound on the extended convex-roof negativity of $\rho\lsc{AB}$:
\begin{equation}\label{CRENbound}
    N^{\mathrm{CREN}}(\rho\lsc{AB})\,\geqslant\,\max\left(\frac{\Tr{\rho\lsc{AB}\,\Pi_V} - \lsup}{2\lsup},\,0\right).
\end{equation}

Noticing that there are at most $m(m-1)/2$ terms in the sum $\sum_{i<j}\,\sqrt{\lph i\lph j}$ and using the Cauchy-Schwarz inequality again, we can write:
\begin{equation}\label{ineqaux}
    \sum_{i<j}\,\sqrt{\lph i\lph j}\,\leqslant\,\sqrt{\frac{m(m-1)}{2}}\sqrt{\sum_{i<j}\lph i\lph j}
\end{equation}
Using Eqs.~(\ref{schconc}), (\ref{concconv}), (\ref{mainbound}), and (\ref{ineqaux}), we obtain a lower bound on the concurrence of an $m\otimes n$ density operator $\rho\lsc{AB}$:
\begin{equation}\label{Concbound}
    C(\rho\lsc{AB})\,\geqslant\,\max\left(\sqrt{\frac{2}{m(m-1)}}\,\frac{\Tr{\rho\lsc{AB}\,\Pi_V} - \lsup}{\lsup},\,0\right).
\end{equation}
When we choose a one-dimensional projector $\Pi_V=\dyad{\Phi}$ on some pure entangled state $\ket{\Phi}$, $\lsup$ is equal to the square of the largest Schmidt coefficient  of $\ket{\Phi}$: $\lsup = \lphf 1$. As an example, if we choose $\Pi_V = \dyad{\Phi^+}$, a one-dimensional projector on the maximally entangled state $$\ket{\Phi^+}=\frac{1}{\sqrt{d}}\sum_{j=0}^{d-1}\ket{jj},$$
then, for the concurrence of a $d\otimes d$ state $\rho$ Eq.~(\ref{Concbound}) gives a bound that was mentioned in Ref.~\cite{ElToSie2015}:
\begin{equation}\label{prevres}
    C(\rho)\,\geqslant\,\max\left(\sqrt{\frac{2d}{d-1}}\left[\bra{\Phi^+}\rho\ket{\Phi^+} - \frac1d\right],\,0\right).
\end{equation}
In addition, when $\Pi_V=\dyad{\Phi}$, a more tight bound than the one  following directly from Eq.~(\ref{Concbound}) can be given for the concurrence: again, we start from the first inequality of Eq.~(\ref{TrBound}):
\begin{multline}\label{pc1}
    \absolutevalue{\Tr{\cvd{\Phi_{\mu}}\cv{\Phi}}}^2\,\leqslant\,\left(\sum_i\,\sqrt{\lphf i\lph i}\right)^2=\\
    =\,\sum_i\,\lph i\lphf i + 2\sum_{i<j}\,\sqrt{\lphf i\lphf j}\sqrt{\lph i\lph j}.
\end{multline}
The first term on the right-hand side is majorized by $\lphf 1$:
\begin{equation}\label{pc2}
\sum_i\,\lph i\lphf i\,\leqslant\,\lphf 1\sum_i\,\lph i = \lphf 1.
\end{equation}
For the second term we use the Cauchy-Swartz inequality and the expression~(\ref{schconc}) for the concurrence of a pure state:
\begin{multline}\label{pc3}
    \sum_{i<j}\,\sqrt{\lphf i\lphf j}\sqrt{\lph i\lph j}\,\leqslant\\
    \leqslant\sqrt{\sum_{i<j}\lphf i\lphf j}\sqrt{\sum_{i<j}\lph i\lph j} = \frac14\,C(\Phi)C(\Phi_{\mu}).
\end{multline}
Eq.~(\ref{Trmatr}) along with Eqs.~(\ref{pc1})-(\ref{pc3}) yields:
\begin{equation}\label{sb1}
    C(\rho\lsc{AB})\,\geqslant\,\max\left(\frac{2\left[\bra{\Phi}\rho\lsc{AB}\ket{\Phi} - \lphf 1\right]}{C(\Phi)},\,0\right),
\end{equation}
for any entangled pure state $\ket{\Phi}$.

For comparison, in this case Eq.~(\ref{Concbound}) would have given a bound:
\begin{equation}\label{sb2}
    C(\rho\lsc{AB})\,\geqslant\,\max\left(\sqrt{\frac{2}{m(m-1)}}\frac{\bra{\Phi}\rho\lsc{AB}\ket{\Phi} - \lphf 1}{\lphf 1},\,0\right).
\end{equation}
The bound in Eq.~(\ref{sb1}) has advantage over the one in Eq.~(\ref{sb2}) only when the state $\rho\lsc{AB}$ is entangled and its entanglement is detected by violation of the condition: $\bra{\Phi}\rho\lsc{AB}\ket{\Phi} \leqslant \lphf 1$.

Since CREN and the concurrence are invariant under local unitaries $U_A$ and $U_B$ , the inequalities in Eqs.~(\ref{CRENbound}) and (\ref{Concbound}) can be optimized over all such transformations:
\begin{subequations}
\begin{eqnarray}
	N^{\mathrm{CREN}}(\rho\lsc{AB})\,&\geqslant&\,\max\left(\frac{{\cal F}_V(\rho\lsc{AB}) - \lsup}{2\lsup},\,0\right),\label{opta}\\
	C(\rho\lsc{AB})\,&\geqslant&\,\max\left(\sqrt{\frac{2}{m(m-1)}}\,\frac{{\cal F}_V(\rho\lsc{AB}) - \lsup}{\lsup},\,0\right),\label{optb}
\end{eqnarray}
\end{subequations}
where
\begin{equation}
   {\cal F}_V(\rho\lsc{AB}) = \max_{U_A,\,U_B}\,\Tr{(U_A\otimes U_B)\rho\lsc{AB}(U_A\otimes U_B)^{\dagger}\Pi_V} 
\end{equation}
is the generalization of the \emph{fully entangled fraction} introduced in Ref.~\cite{EntFrac}.

\section{Applications}
\subsection{Bounds for some states}

In this subsection we use the derived inequalities to calculate the bounds for some well-known states.

For isotropic states,
\begin{equation}
\rho_F=\frac{1-F}{d^2-1}\left(I-\ket{\Phi^+}\bra{\Phi^+}\right)
+F\ket{\Phi^+}\bra{\Phi^+},\label{eq:isotropic}
\end{equation}
the lower bounds for CREN and the concurrence, obtained from Eqs.~(\ref{CRENbound}) and (\ref{Concbound}) with $\Pi_V$ equal to $\dyad{\Phi^+}$, are easily calculated:
\begin{subequations}\label{MIso}
\begin{eqnarray}
	N^{\mathrm{CREN}}(\rho_F)&\geqslant&\max\left(\frac{Fd - 1}2,\,0\right),\label{isoa}\\
	C(\rho_F) \,&\geqslant&\, \max\left(\sqrt{\frac{2d}{d-1}}\left(F-1/d\right),\,0\right)\label{isob}.
\end{eqnarray}
\end{subequations}
For Werner states,
\begin{eqnarray}
\varrho_W&=& \frac{2(1-W)}{d(d+1)}
\left(\sum_{k=0}^{d-1}\dyad{kk}+\sum_{i<j}\dyad*{\Psi_{ij}^+}\right)
\nonumber \\
&&+\frac{2W}{d(d-1)}\sum_{i<j}\dyad*{\Psi_{ij}^-},\label{eq:Werner}
\end{eqnarray}
with 
\begin{equation*}
    \ket{\Psi_{ij}^{\pm}}=\left(\ket{ij}\pm\ket{ji}\right)/{\sqrt{2}},
\end{equation*}
to  ``extract'' the parameter $W$, we need to consider the projector on the anti-symmetric subspace:
$$
\Pi_V = \sum_{i<j}\dyad{\Psi_{ij}^-}.
$$
In this case,
$$
    \Tr{\varrho_W\Pi_V} = W.
$$
As for the supremum $\lsup$, we state the following
\begin{lemma}
For vector states in the anti-symmetric subspace of the Hilbert space the supremum of the largest Schmidt coefficient squared, $\lsup$, is equal to $\frac 12$.
\end{lemma}
\begin{proof}
	We can notice that to each linear combination of the vector states $\{\Psi_{ij}^-\}$ corresponds some anti-symmetric matrix~(defined in the same way as those  in Eq.~(\ref{PureDec})). According to Ref.~\cite{MAFull}, the nonzero singular values of an arbitrary anti-symmetric matrix $A$ are $\sigma_1(A),\,\sigma_1(A),\,\ldots,\,\sigma_{r/2}(A),\, \sigma_{r/2}(A)$, where $r = \mathrm{rank}\,A$ is always even; each nonzero singular value is repeated twice.
Consequently, for any vector state $\ket{\Psi}$ in the anti-symmetric subspace $\lps 1 \leqslant 1/2$; otherwise, there would be two coefficients, $\lps 1,\,\lps 2$, such that: $\lps 1 = \lps 2 > 1/2$, which is impossible since their sum cannot exceed 1. The upper bound, $\lsup = 1/2$, is achieved, for example, on the vector states $\{\Psi_{ij}^-\}$ themselves.
\end{proof}

Using this result, from Eqs.~(\ref{CRENbound}) and (\ref{Concbound}) we obtain:
\begin{subequations}\label{MWerner}
\begin{eqnarray}
	N^{\mathrm{CREN}}(\varrho_W)\,&\geqslant&\,\max\left(\frac{2W-1}{2},\,0\right),\label{werna}\\
	C(\varrho_W)\,&\geqslant&\,\max\left(\sqrt{\frac{2}{d(d-1)}}\,(2W-1),\,0\right).
\end{eqnarray}
\end{subequations}
The lower bounds of Eqs.~(\ref{isoa}), (\ref{isob}), and (\ref{werna}) turn out to be the exact CREN~\footnote{In Ref.~\cite{CREN2003} there is an additional factor, $2/(d-1)$, in the definition of the negativity of a pure state, which should be taken into account} and concurrence values for isotropic and Werner states obtained in Refs.~\cite{CREN2003} and \cite{CIso2003}.

Consider a density operator
\begin{equation}
    \rho = \frac{2F}{d(d-1)}\sum_{i<j}\dyad{\Psi_{ij}^-} + (1-F)\dyad{\Phi^+}.
\end{equation}
Combination of convexity property with Eq.~(\ref{CRENbound}) can give quite informative upper and lower bounds on the negativity. Applying Eq.~(\ref{CRENbound}) with $\Pi_V = \sum_{i<j}\dyad{\Psi_{ij}^-}$ and then with $\Pi_V = \dyad{\Phi^+}$, we obtain:
$$
  N^{\mathrm{CREN}}(\rho) \geqslant\max\left(F-1/2,\,\frac12\left[d(1-F)-1\right],\,0\right).
$$
By convexity of CREN,
\begin{multline*}
    N^{\mathrm{CREN}}(\rho)\leqslant\frac{2F}{d(d-1)}\,\sum_{i<j}\,N^{\mathrm{CREN}}
    \left(\dyad{\Psi_{ij}^-}\right)+\\
    + (1-F)\,N^{\mathrm{CREN}}\left(\dyad{\Phi^+}\right)=\,\frac12\left(d-1 - (d-2)F\right),
\end{multline*}
where we have used known expressions:
$$
N^{\mathrm{CREN}}
    \left(\dyad{\Psi_{ij}^-}\right) = \frac12,
$$
and 
$$
N^{\mathrm{CREN}}\left(\dyad{\Phi^+}\right) = \frac{d-1}{2}.
$$
Combining all bounds, we have:
\begin{multline}\label{BoundEx}
   \frac12\left(d-1 - (d-2)F\right)\,\geqslant\,N^{\mathrm{CREN}}(\rho)\,\geqslant\\
   \geqslant\,\left\{\begin{aligned}
   \frac12\left(d(1-F)-1\right),&&\:0\leqslant F\leqslant\frac{d}{d+2}\\
   F-1/2,&&\:\frac{d}{d+2}\leqslant F\leqslant 1.
   \end{aligned}\right.
\end{multline}
When $d>2$, the lower bound is always positive, and $\rho$ is entangled. When $d=2$, from the PPT criterion it follows that the state is separable at $F = 1/2$, and  the lower bound of Eq.~(\ref{BoundEx}) gives exact CREN values for three points: $F = 0,\,1/2,\,1$. Since there is no larger convex function with graph coming through these three points, the lower bound of Eq.~(\ref{BoundEx}) coincides with the exact CREN value in this case:
$$
N^{\mathrm{CREN}}(\rho) = \absolutevalue{F-1/2}\:\:\:\mbox{for d = 2}.
$$

\subsection{Separability criterion}

From Eq.~(\ref{CRENbound}) (or Eq.~(\ref{Concbound})) the following condition can be obtained:

\emph{If $\rho\lsc{AB}$ is separable then for a projector $\Pi_V$ on some subspace $V$ of $\hil_A\otimes\hil_B$}
\begin{equation}
    \Tr{\rho\lsc{AB}\,\Pi_V}\,\leqslant\,\lsup.\label{sepmain}
\end{equation}
This criterion, along with its remarkable consequences, was derived earlier in Ref.~\cite{SpecProp} with the use of the theory of entanglement witnesses.
One interesting consequence of this criterion:

\emph{If $\rho\lsc{AB}$ is separable then for any its ensemble decomposition $\rho_{\scr AB}=\sum_{\mu} q_{\mu} \dyad{\Phi_{\mu}}\lsc{AB}$ the following inequality holds:
\begin{equation}
	\lph 1\,\geqslant\,q_{\mu}.
\end{equation}}

\subsection{Robustness of entanglement}
Let us suppose that the entanglement of some $m\otimes n$ state $\rho\lsc {AB}$ is detected by violation of inequality~(\ref{sepmain}):
\begin{equation}\label{delt}
	\Tr{\rho\lsc{AB}\,\Pi_V}\,=\,\lsup + \delta,\qquad \delta > 0. 
\end{equation}
Since $\lsup$ is fixed for  the chosen subspace $V$, it is convenient to estimate how stable the  entanglement of the state is under various perturbations.

We begin with analyzing general Hermitian perturbations $\Delta$ satisfying 
\begin{equation}\label{stcond}
	\Tr{\Delta} = 0;\qquad  \rho\lsc {AB} + \Delta\,\geqslant\,0,
\end{equation}
i.~e., those $\Delta$ for which $\rho\lsc {AB} + \Delta$ remains  to be the state. Let $k$ -- dimension of the subspace $V$. Using the von Neumann inequality~(\ref{TrSing}) for Hermitian operators $\Pi_V$ and $\Delta$ and the fact that $\Pi_V$ has $k$ eigenvalues equal to one and the rest eigenvalues equal to zero, we obtain:
\begin{equation}\label{dbound}
	\absolutevalue{\Tr{\Pi_V\Delta}}\,\leqslant\,\norm{\Delta}_{(k)},
\end{equation}
where $\norm{A}_{(k)}$ -- the \emph{Ky Fan $k$-norm}~\cite{BhFan} of a matrix $A$, the sum of $k$ largest singular values of $A$:
\begin{equation}
	\norm{A}_{(k)} = \sum_{i=1}^k\,\sigma_i (A).
\end{equation}
When $A$ is Hermitian~(as with $\Delta$), the norm turns into the sum of $k$ largest absolute values of the eigenvalues of $A$.

The same reasoning applied to Eq.~(\ref{delt}) yields an upper bound on $\delta$ itself:
\begin{equation}
	\delta\,\leqslant\,\norm{\rho\lsc{AB}}_{(k)} - \lsup.
\end{equation}
Combining Eqs.~(\ref{delt}) and (\ref{dbound}), we obtain the following result:

\emph{For perturbations $\Delta$ satisfying, in addition to Eq.~(\ref{stcond}), the condition 
\begin{equation}
	\norm{\Delta}_{(k)}\,<\,\delta,
\end{equation}
the state $\rho\lsc {AB} + \Delta$ is entangled.} 

As a more physical example, we can consider mixing taking place between an entangled state $\rho$ and another state $\rho_M$. Following Refs.~\cite{rob},~\cite{robgen}, we define robustness of $\rho$ relative to $\rho_M$ as the minimal $p\in[0;\,1]$ for which $(1-p)\rho + p\rho_M$ is separable. We will investigate ``robustness from spectrum'' ---  robustness of $\rho$ relative to states $\rho_M$ with some given information about their spectrum.
The following property~\cite{MATr} is important in such an analysis:

\emph{Let $A$ and $B$ be Hermitian and have respective vectors of eigenvalues $\lambda(A) = [\lambda_i(A)]_{i=1}^n$ and $\lambda(B) = [\lambda_i(B)]_{i=1}^n$. Then
\begin{equation}
	\sum_{i=1}^n\,\lambda_i (A)^{\downarrow}\,\lambda_i (B)^{\uparrow}\,\leqslant\,\tr{AB}\,\leqslant\,\sum_{i=1}^n\,\lambda_i (A)^{\downarrow}\,\lambda_i (B)^{\downarrow},
\end{equation}
where $\uparrow$ and $\downarrow$ denote increasing and decreasing ordering of lambda's respectively.}

Using the above property, we obtain:
\begin{equation}
	\sum_{i=1}^n\,\lambda_i (\rho_M)^{\uparrow}\,\leqslant\,\tr{\Pi_V\rho_M}\,\leqslant\,\sum_{i=1}^n\,\lambda_i (\rho_M)^{\downarrow}.
\end{equation}
Due to the fact that $\tr{\rho_M} = 1$ the last expression can be rewritten in terms of Ky Fan norms:
\begin{equation}\label{Fbounds}
	1-\norm{\rho_M}_{(mn-k)}\,\leqslant\,\tr{\Pi_V\rho_M}\,\leqslant\,\norm{\rho_M}_{(k)}.
\end{equation}
Let us assume that equality in Eq.~(\ref{delt}) holds for $\rho$. Making use of Eq.~(\ref{Fbounds}), we obtain:
\begin{equation}
	\tr{\Pi_V \left((1-p)\rho + p\rho_M\right)} - \lsup\,\geqslant\,\delta + p(1 - \norm{\rho_M}_{(mn-k)} - \lsup - \delta).
\end{equation}
From the last expression it follows that \emph{for $p$ satisfying
\begin{equation}\label{FinB}
	p\,<\,\frac {\delta}{\lsup + \delta + \norm{\rho_M}_{(mn-k)} -1}
\end{equation}
the state $(1-p)\rho + p\rho_M$ is entangled.}

Eq.~(\ref{FinB}) gives a lower bound on robustness of entanglement of $\rho$ relative to states $\rho_M$ with the given Ky Fan norm $\norm{\rho_M}_{(mn-k)}$.

In a particular case when subspace $V$ coincides with the span of eigenvectors of $\rho$  corresponding to nonzero eigenvalues and $\lsup<1$, we have: $$\lsup + \delta = \tr{\Pi_V\rho} = 1; \qquad \delta = 1 - \lsup,$$
and Eq.~(\ref{FinB}) transforms into
\begin{equation}\label{redp}
	p\,<\,(1-\lsup)\,\norm{\rho_M}_{(mn-k)}^{-1}.
\end{equation}
\emph{Example 1}. If we choose the Bell state, $\rho = \dyad*{\Phi^+}$, then $\lsup = 1/d$, $k = 1$~($V$ is a one dimensional subspace spanned by $\ket{\Phi^+}$), and from Eq.~(\ref{redp}) we obtain that entanglement is preserved under mixing when $$p < \frac {d-1}{d\,\norm{\rho_M}_{(d^2-1)}} = \frac {d-1}{d(1-\lambda_{\min}(\rho_M))}.$$
One may try to increase the minimal eigenvalue of the mixing noise as much as possible to achieve greater robustness of entanglement in this case.
\medskip

\noindent
\emph{Example 2}. For a state $$\rho = \frac{2}{d(d-1)}\sum_{i<j}\dyad*{\Psi_{ij}^-},$$ which is in fact a Werner state of Eq.~(\ref{eq:Werner}) corresponding to $W=1$, we have: $\lsup = 1/2$, $k=d(d-1)/2$, and the bound for probability of mixing $p$ is:
$$
p\,<\,\frac12 \norm{\rho_M}_{(\scriptscriptstyle d(d+1)/2)}^{-1}.
$$
\section{Conclusions}

Our main results are the inequalities in Eqs.~(\ref{CRENbound}), (\ref{Concbound}), (\ref{sb1}) and their optimized over local unitaries versions - Eqs.~(\ref{opta}) and (\ref{optb}). Some of their consequences generalize previously known results: Eq.~(\ref{sb1}) is an extension of the estimate in Eq.~(\ref{prevres}) known from the Schmidt number witness~\cite{SchWit}.

In general, the supremum of the largest Schmidt coefficient squared, $\lsup$, is hard to evaluate  for an arbitrary many-dimensional subspace $V$: direct calculations of singular values and a maximization procedure over a large number of parameters are needed. In some special cases it can be obtained from general results of matrix theory -   we did this for the case of  the anti-symmetric subspace. We applied this result to  Werner states, and  our bound on the convex-roof extended negativity gave the best possible result -  the exact value of this measure. It would be interesting to find other examples of subspaces with  relatively low values of $\lsup$.

The von Neumann's trace inequality played a crucial role in derivation of our results: it allowed us to relate the Schmidt coefficients with the expectation values of specific operators. An interesting direction of further research would be to analyze other trace inequalities and properties of singular values which could potentially give some information about entanglement measures.

\section*{Acknowledgments}

The author is grateful to G. Sarbicki for kind explaining the details of his work. The author would like to thank A. Khalapyan for comments on an earlier version of the manuscript. This work was supported by Lomonosov Moscow State University.

\bibliography{bounds}

\end{document}